\documentclass{article}
\usepackage{fullpage}
\usepackage{graphicx}
\usepackage{amsmath, amsthm, amssymb, cite, color, comment, mathtools, url}
\usepackage[unicode]{hyperref}

\newtheorem{theorem}{Theorem}[section]
\newtheorem{lemma}[theorem]{Lemma}

\newtheorem{claim}[theorem]{Claim}

\theoremstyle{definition}

\newtheorem*{problem}{Problem}

\theoremstyle{remark}
\newtheorem{remark}[theorem]{Remark}

\newcommand{\Ain}{A_\mathrm{in}}
\newcommand{\Aout}{A_\mathrm{out}}
\newcommand{\Yin}{\tilde{Y}_\mathrm{in}}
\newcommand{\Yout}{\tilde{Y}_\mathrm{out}}
\newcommand{\Zin}{\tilde{Z}_\mathrm{in}}
\newcommand{\Zout}{\tilde{Z}_\mathrm{out}}
\newcommand{\Hin}{\tilde{H}_\mathrm{in}}
\newcommand{\Sopt}{S_\mathrm{opt}}
\newcommand{\Mopt}{M_\mathrm{opt}}

\title{Approximation and FPT Algorithms\\for Finding DM-Irreducible Spanning Subgraphs}
\author{
  Ryoma Norose\thanks{Independent Researcher (Graduated from Osaka University on March 2024), Japan.} \and
  Yutaro Yamaguchi\thanks{Osaka University, Japan. Email: \texttt{yutaro.yamaguchi@ist.osaka-u.ac.jp}}}
\date{\empty}

\begin{document}
\maketitle
\thispagestyle{empty}

\begin{abstract}
Finding a minimum-weight strongly connected spanning subgraph of an edge-weighted directed graph is equivalent to the weighted version of the well-known strong connectivity augmentation problem.
This problem is NP-hard, and a simple $2$-approximation algorithm was proposed by Frederickson and J{\'a}j{\'a} (1981); surprisingly, it still achieves the best known approximation ratio in general.
Also, Bang-Jensen and Yeo (2008) showed that the unweighted problem is FPT (fixed-parameter tractable) parameterized by the difference from a trivial upper bound of the optimal value.
In this paper, we consider a generalization related to the Dulmage--Mendelsohn decompositions of bipartite graphs instead of the strong connectivity of directed graphs, and extend these approximation and FPT results to the generalized setting.
\end{abstract}


\clearpage
\setcounter{page}{1}

\section{Introduction}\label{section:1}
A directed graph is called \emph{strongly connected} if there exists a path between any pair of vertices.
The strong connectivity augmentation problem (\textsc{SCA}) is, given a directed graph and the weights of additional edges, to find a minimum-weight subset of additional edges making the graph strongly connected.
While \textsc{SCA} is NP-hard in general, it can be solved in linear time in the unweighted case (when an edge between any pair of veritces can be added with the same weight)~\cite{eswaran1976augmentation}.

The following problem, called the strongly connected spanning subgraph problem (\textsc{SCSS}), is polynomial-time equivalent to \textsc{SCA} in general: given a strongly connected directed graph with edge weight, find a minimum-weight spanning subgraph that is strongly connected.
This problem is one of the most fundamental special cases of the survivable network design problem~\cite{steiglitz1969design, dahl1993design}.
Frederickson and J{\'a}j{\'a}~\cite{frederickson1981approximation} gave a $2$-approximation algorithm for \textsc{SCSS}; this is based on a very simple observation (see Section~\ref{subsection:3.1}) and, surprisingly, it still achieves the best known approximation ratio in general.
For the unweighted case (when all the edges have the same weight, which is already NP-hard in contrast to \textsc{SCA}), several algorithms with better approximation ratios \cite{khuller1996strongly,vetta2001approximating,zhao2003linear,khuller1995approximating} have been proposed.
Also, Bang-Jensen and Yeo~\cite{bang2008minimum} showed that the unweighted case is FPT parameterized by the difference from a trivial upper bound of the optimal value (see Sections~\ref{section:2.2} and \ref{section:4.1});
here, the problem is called \emph{FPT (fixed-parameter tractable)} if it can be solved in $f(k) \cdot n^{O(1)}$ time, where $k$ and $f$ are the prescribed parameter and computable function, respectively, and $n$ is the input size.

\textsc{SCSS} (or \textsc{SCA}) is generalized to several directions.
One natural generalization is the $k$(-edge)-connected spanning subgraph (or $k$(-edge)-connectivity augmentation) problem, where a directed graph is called \emph{$k$-connected} (or \emph{$k$-edge-connected}) if there exists $k$ openly vertex-disjoint (or edge-disjoint) paths between any pair of vertices.
The case when $k = 1$ corresponds to \textsc{SCSS} (or \textsc{SCA}), and the general case is also well-studied~\cite{frank1992augmenting,frank1995minimal,khuller1994biconnectivity,nutov2014approximating,bang2018parameterized}.

B\'{e}rczi, Iwata, Kato, and Yamaguchi~\cite{berczi2018DMImaking} recently introduced another generalization of \textsc{SCA}, which is related to the \emph{Dulmage--Mendelsohn decomposition} (the \emph{DM-decomposition}) \cite{dulmage1958coverings,dulmage1959structure} of a bipartite graph.
The DM-decomposition gives a partition of the vertex set that reflects the structure of all maximum matchings therein, and a bipartite graph is \emph{DM-irreducible} if its DM-decomposition consists of a single component.
This property is closed under adding edges, and we analogously formulate the DM-irreducible spanning subgraph problem (\textsc{DMISS}): given a DM-irreducible bipartite graph with edge weight, find a minimum-weight spanning subgraph that is DM-irreducible.
In terms of \textsc{DMISS}, B\'{e}rczi et al.~\cite{berczi2018DMImaking} proposed polynomial-time algorithms for the case when the input bipartite graph is unbalanced (which reduces to the weighted matroid intersection problem), and for the case when the input graph is a balanced complete bipartite graph and the weight of each edge is $0$ or $1$; here a bipartite graph is \emph{balanced} if the two color classes have the same number of vertices.
\textsc{DMISS} for balanced bipartite graphs generalizes \textsc{SCSS} (cf.~\cite[Section 2.3]{berczi2018DMImaking}), and the latter result extends the polynomial-time solvability of the unweighted case of \textsc{SCA}~\cite{eswaran1976augmentation}.

In this paper, we investigate \textsc{DMISS} for balanced bipartite graphs, whose algorithmic difficulty has remained widely open (other than the fact that it is NP-hard even in the unweighted case as so is SCSS).
As natural but nontrivial results, by extending the approaches to SCSS by Frederickson and J{\'a}j{\'a}~\cite{frederickson1981approximation} and by Bang-Jensen and Yeo~\cite{bang2008minimum}, we propose a $2$-approximation algorithm in general and an FPT algorithm for the unweighted case.
On the approximation result, we first show that a naive extension of \cite{frederickson1981approximation} gives a $3$-approximation algorithm, and then improve it to a $2$-approximation algorithm with the aid of the matroid intersection framework.
On the FPT result, we demonstrate that the idea of \cite{bang2008minimum} can be extended by carefully adjusting several key concepts and discussions.

The rest of this paper is organized as follows.
In Section \ref{section:2}, we describe necessary definitions, and formally state the problems and our results.
In Section~\ref{section:3}, we give a $2$-approximation algorithm for \textsc{DMISS}.
In Section~\ref{section:4}, we prove the unweighted case of \textsc{DMISS} is FPT parameterized by the difference from a trivial upper bound of the optimal value.

\section{Preliminaries}\label{section:2}
\subsection{Definitions}
We refer the readers to \cite{schrijver2003combinatorial} for basic concepts and notation on graphs.
In this paper, a \emph{graph} means a directed graph or an undirected graph.
We call a directed graph simply a \emph{digraph}.

For a graph $G$, we denote by $V(G)$ and $E(G)$ the vertex set and the edge set of $G$, respectively.
A \emph{subgraph} of $G$ is a graph $G'$ such that $V(G') \subseteq V(G)$, $E(G') \subseteq E(G)$, and each edge $e \in E(G')$ is between two vertices in $V(G')$.
A subgraph $G'$ of $G$ is \emph{induced by} a vertex set $X \subseteq V(G)$, denoted by $G[X]$, if $V(G') = X$ and $E(G')$ is the set of edges in $E(G)$ whose end vertices are both in $X$.
We also define $G - X \coloneqq G[V(G) \setminus X]$ for $X \subseteq V(G)$.
A subgraph $G'$ of $G$ is \emph{spanning} if $V(G') = V(G)$.
We often do not distinguish between a spanning subgraph $G'$ and its edge set $E(G')$.
Also, we often specify a graph $G$ as the pair $(V, E)$ of its vertex set $V = V(G)$ and edge set $E = E(G)$.

Let $G = (V, E)$ be an undirected graph.
A \emph{stable set} in $G$ is a vertex set $X \subseteq V$ such that no edge in $E$ has its both end vertices in $X$.
We say that $G$ is \emph{bipartite} if there exists a \emph{bipartition} $(V^+, V^-)$ of $V$ such that $V^+ \cap V^- = \emptyset$, $V^+ \cup V^- = V$, and both $V^+$ and $V^-$ are stable sets in $G$ (possibly $V^+$ or $V^-$ is empty).
To specify the bipartition, we denote by $G = (V^+, V^-; E)$.
For a vertex set $X$ on one side (i.e., $X \subseteq V^+$ or $X \subseteq V^-$), let $\Gamma_G(X)$ denote the \emph{neighborhood} of $X$, i.e., the set of vertices adjacent to some vertex in $X$.
A bipartite graph $G = (V^+, V^-; E)$ is \emph{balanced} if $|V^+| = |V^-|$.

For a digraph $G = (V, E)$, the \emph{underlying graph} is an undirected graph $\overline{G} = (V, \overline{E})$ defined as follows.
For each directed edge $e = (u, v) \in E$, we denote by $\overline{e} = \{u, v\}$ the corresponding undirected edge.
We then define $\overline{F} \coloneqq \{ \overline{e} \mid e \in F \}$ for each $F \subseteq E$.

For each edge $e = \{u, v\}$ in a bipartite graph $G = (V^+, V^-; E)$ such that $u \in V^+$ and $v \in V^-$, we denote its orientations by $\overrightarrow{e} = (u, v)$ and $\overleftarrow{e} = (v, u)$.
Similarly, for each $F \subseteq E$, we define $\overrightarrow{F}$ and $\overleftarrow{F}$ as the sets of the orientations $\overrightarrow{e}$ and $\overleftarrow{e}$, respectively, of the edges $e \in F$.

An undirected graph is said to be \emph{connected} if for each pair of vertices, there exists a path between them.
A digraph is said to be \emph{strongly connected} if the same condition (where a path means a directed path) holds, and just \emph{connected} if the underlying graph is connected.

For an undirected graph $G$ and a vertex $v \in V(G)$, we denote by $\deg_G(v)$ the \emph{degree} of $v$, which is the number of edges incident to $v$.
Similarly, for a digraph $G$ and a vertex $v \in V(G)$, we denote by $\deg_G^\mathrm{in}(v)$ and $\deg_G^\mathrm{out}(v)$ the \emph{in-degree} and \emph{out-degree} of $v$, respectively, which are the numbers of edges entering and leaving $v$.
An \emph{in-arborescence} (or \emph{out-arborescence}) is a connected digraph in which exactly one vertex, called \emph{the root}, has out-degree (or in-degree, respectively) $0$ and all other vertices have out-degree (or in-degree, respectively) $1$.
An in- or out-arborescence having its root $r$ is said to be \emph{rooted at} $r$, and is also called an \emph{$r$-in-arborescence} or an \emph{$r$-out-arborescence}, respectively.

A \emph{matching} in an undirected graph $G$ is an edge subset in which no two edges share a vertex.
If a matching covers all the vertices in $G$, it is called a \emph{perfect matching}.
For a perfect matching $M$ in a balanced bipartite graph $G = (V^+, V^-; E)$, the \emph{auxiliary graph} of $G$ with respect to $M$ is a digraph $G_{M} \coloneqq (V; \overrightarrow{E} \cup \overleftarrow{M})$.

The \emph{DM-decomposition} \cite{dulmage1958coverings,dulmage1959structure} of a bipartite graph is formally defined by the distributive lattice formed by the minimizers of a submodular function, and a bipartite graph is called \emph{DM-irreducible} if its DM-decomposition consists of a single component.
We here skip the formal definition and instead employ its characterization as an alternative ``definition'' (cf.~\cite[Section~2.2]{berczi2018DMImaking}).\footnote{The equivalent condition is also called \emph{elementary}, e.g., in \cite{lovasz1986matching,frank2011connections}.}

\begin{lemma}[cf.~{\cite[Theorem~2.4.18]{frank2011connections}}]\label{lemma: DM-irreducible}
For a balanced bipartite graph $G$, the following statements are equivalent.
\begin{enumerate}
\setlength{\itemsep}{0mm}
\item $G$ is DM-irreducible.
\item $G$ has a perfect matching $M$ such that $G_M$ is strongly connected.
\item $G$ has a perfect matching and $G_M$ is strongly connected for any perfect matching $M$ in $G$.
\end{enumerate}
\end{lemma}

Suppose that a graph $G$ is associated with edge weight $w \colon E(G) \to \mathbb{R}$.
We define the \emph{weight} of each edge set $F \subseteq E(G)$ as $w(F) \coloneq \sum_{e \in F} w(e)$.

\subsection{Problems and Our Results}\label{section:2.2}
We formally state the problems dealt with in this paper.
Recall that \textsc{DMISS} for unbalanced bipartite graphs is polynomial-time solvable~\cite{berczi2018DMImaking}, which is not considered here.

\begin{problem}[\textsc{Strongly Connected Spanning Subgraph (SCSS)}]
\mbox{ }\vspace{-1mm}
\begin{description}
\setlength{\itemsep}{0mm}
\item[Input:] A strongly connected digraph $G$ with $w \colon E(G)  \to \mathbb{R}_{\ge 0}$.
\item[Goal:] Find a minimum-weight strongly connected spanning subgraph of $G$.
\end{description}
\end{problem}

\smallskip

\begin{problem}[\textsc{DM-Irreducible Spanning Subgraph (DMISS)}]
\mbox{ }\vspace{-1mm}
\begin{description}
\setlength{\itemsep}{0mm}
\item[Input:] A DM-irreducible balanced bipartite graph $G$ with $w \colon  E(G)  \to \mathbb{R}_{\ge 0}$.
\item[Goal:] Find a minimum-weight DM-irreducible spanning subgraph of $G$.
\end{description}
\end{problem}

It is easy to observe that any inclusion-wise minimal strongly connected spanning subgraph of a digraph on $n$ vertices consists of at most $2n - 2$ edges (by Lemma~\ref{lemma: SC includes arborescence}).
Similarly, any inclusion-wise minimal DM-irreducible spanning subgraph of a balanced bipartite graph on $2n$ vertices ($n$ vertices on each side) consists of at most $3n - 2$ edges (by Lemma~\ref{lemma: DMI includes arborescence}).
The following problems are the unweighted version of the above two problems parameterized by the difference from these trivial upper bounds; note that it is obvious from the definition that $|V(G)|$ edges are necessary in either problem.

\begin{problem}[\textsc{UnweightedSCSS}]
\mbox{ }\vspace{-1mm}
\begin{description}
\setlength{\itemsep}{0mm}
\item[Input:] A strongly connected digraph $G$ on $n$ vertices.
\item[Parameter:] An integer $k \le n - 2$.
\item[Goal:] Test whether $G$ has a strongly connected spanning subgraph with at most $2n - 2 - k$ edges or not.
\end{description}
\end{problem}

\smallskip

\begin{problem}[\textsc{UnweightedDMISS}]
\mbox{ }\vspace{-1mm}
\begin{description}
\setlength{\itemsep}{0mm}
\item[Input:] A DM-irreducible balanced bipartite graph $G$ on $2n$ vertices.
\item[Parameter:] An integer $k \le n - 2$.
\item[Goal:] Test whether $G$ has a DM-irreducible spanning subgraph with at most $3n - 2 - k$ edges or not.
\end{description}
\end{problem}

The results of this paper are summarized as follows, where $n = \frac{1}{2}|V(G)|$.

\begin{theorem}\label{thm:approx}
    There exists a $2$-approximation algorithm for \textsc{DMISS} that runs in $O(n^3)$ time.
\end{theorem}

\begin{theorem}\label{thm:FPT}
    There exists an FPT algorithm for \textsc{UnweightedDMISS}. 
\end{theorem}

\section{Approximation Algorithms for \textsc{DMISS}}\label{section:3}
In this section, we prove Theorem~\ref{thm:approx} by presenting a $2$-approximation algorithm for \textsc{DMISS} that runs in $O(n^3)$ time.

In Section~\ref{subsection:3.1}, as a base of our result, we review the $2$-approximation algorithm for \textsc{SCSS} proposed by Frederickson and J{\'a}j{\'a}~\cite{frederickson1981approximation}.
In Section~\ref{subsection:3.2}, we show a $3$-approximation algorithm for \textsc{DMISS} obtained by a naive extension of the result of Frederickson and J{\'a}j{\'a}.
In Section~\ref{subsection:3.3}, we improve it to a $2$-approximation algorithm with the aid of matroid intersection.

\subsection{2-approximation Algorithm for \textsc{SCSS}}\label{subsection:3.1}
We start with an elementary characterization of strongly connected digraphs.

\begin{lemma}\label{lemma: SC includes arborescence}
For a digraph $G$, the following statements are equivalent.
\begin{enumerate}
\setlength{\itemsep}{0mm}
\item $G$ is strongly connected.
\item For some vertex $r$, $G$ has a spanning in-arborescence and out-arborescence rooted at $r$.
\item For any vertex $r$, $G$ has a spanning in-arborescence and out-arborescence rooted at $r$.
\end{enumerate}
\end{lemma}

For the input $(G, w)$ of \textsc{SCSS}, the algorithm arbitrarily chooses a vertex $r$ and finds a minimum-weight spanning in-arborescence $\Ain$ and out-arborescence $\Aout$ rooted at $r$; the existence of $\Ain$ and $\Aout$ follows from Lemma~\ref{lemma: SC includes arborescence} $(1 \Rightarrow 3)$.
This can be done in $O(m + n \log n)$ time~\cite{gabow1986efficient}, where $n = |V(G)|$ and $m = |E(G)|$.
By Lemma \ref{lemma: SC includes arborescence} $(2 \Rightarrow 1)$, the union of $\Ain$ and $\Aout$ results in a strongly connected spanning subgraph of $G$, and the algorithm outputs this.
Since an optimal solution has a spanning in-arborescence and out-arborescence rooted at $r$ by Lemma \ref{lemma: SC includes arborescence} $(1 \Rightarrow 3)$ and their weights are at least those of $\Ain$ and $\Aout$, respectively, the weight of the output solution is at most twice the optimal value.

\subsection{3-approximation Algorithm for \textsc{DMISS}}\label{subsection:3.2}
A naive extension of the above result leads to a $3$-approximation algorithm for \textsc{DMISS}.
First, Lemma \ref{lemma: SC includes arborescence} is straightforwardly extended as follows (recall Lemma~\ref{lemma: DM-irreducible}).

\begin{lemma} \label{lemma: DMI includes arborescence}
For a balanced bipartite graph $G$, the following statements are equivalent.
\begin{enumerate}
\setlength{\itemsep}{0mm}
\item $G$ is DM-irreducible.
\item $G$ has a perfect matching $M$ and a vertex $r$ such that $G_M$ has a spanning in-arborescence and out-arborescence rooted at $r$.
\item $G$ has a perfect matching, and for any perfect matching $M$ and any vertex $r$, $G_M$ has a spanning in-arborescence and out-arborescence rooted at $r$.
\end{enumerate}
\end{lemma}

To achieve $3$-approximation for \textsc{DMISS}, we simply divide the problem into two tasks: finding a minimum-weight perfect matching $M$ in $(G, w)$ and solving \textsc{SCSS} with the input $(G_M, w_M)$, where $w_M \colon E(G_M) \to \mathbb{R}_{\geq 0}$ is defined as
\begin{align}
w_M(e) \coloneqq \begin{cases}
w(\overline{e}) & (\overline{e} \not\in M),\\
0 & (\overline{e} \in M).
\end{cases}\label{eq:w_M}
\end{align}
For the first task, polynomial-time exact algorithms initiated by the Hungarian method~\cite{kuhn1955hungarian} are well-known.
For the second task, we apply the $2$-approximation algorithm described in the previous section.
We then output the union of the solutions (for the second solution, we take the set of corresponding edges in $G$).
This intuitively leads to a $3$-approximation algorithm for \textsc{DMISS}, which can be formally shown as follows.

\begin{lemma}\label{lemma: 3-approx}
There exists a polynomial-time $3$-approximation algorithm for \textsc{DMISS}.
\end{lemma}
\begin{proof}
Let $M$ be a minimum-weight perfect matching in $(G, w)$ (for the first task), and $\Ain$ and $\Aout$ be a minimum-weight spanning in-arborescence and out-arborescence in $(G_M, w_M)$ (for the second task), respectively.
Also, let $\Sopt$ be an optimal solution and $\Mopt$ be a perfect matching in $\Sopt$. 
We then have $w(M)\le w(\Mopt) \le w(\Sopt)$.

Consider a subgraph $G' = (V^+, V^-; \Sopt \cup M)$ of $G$, which is DM-irreducible as so is $\Sopt$.
By Lemma~\ref{lemma: DMI includes arborescence}, the auxiliary graph $G'_M$ has a spanning in-arborescence and out-arborescence rooted at any vertex.
Let $r$ be the root of $\Ain$ and $\Aout$, and $\Ain'$ and $\Aout'$ be any spanning in-arborescence and out-arborescence of $G'_M$ rooted at $r$, respectively.
We then have
\begin{align*}
    w((\overline{\Ain} \cup \overline{\Aout}) \setminus M) &= w_M(\Ain \cup \Aout) && \text{(definition of $w_M$)}\\
    &\le w_M(\Ain) + w_M(\Aout) && \text{($w_M$ is nonnegative)}\\
    &\le w_M(\Ain') + w_M(\Aout') && \text{($G'_M$ is a subgraph of $G_M$)}\\
    &= w(\overline{\Ain'} \setminus M) + w(\overline{\Aout'} \setminus M) && \text{(definition of $w_M$)}\\
    &\le 2w(\Sopt) && (E(G') \setminus M \subseteq \Sopt).
\end{align*}
Thus, the weight of the output solution is 
\[w(M \cup \overline{\Ain} \cup \overline{\Aout}) = w(M) + w((\overline{\Ain} \cup \overline{\Aout}) \setminus M) \le 3w(\Sopt). \qedhere\]
\end{proof}

\subsection{2-approximation Algorithm for DMISS}\label{subsection:3.3}
We improve the above $3$-approximation algorithm with the aid of matroid intersection.

We first give an easy but important observation.

\begin{lemma}\label{lemma:matching in arborescence}
    Let $G = (V^+, V^-; E)$ be a DM-irreducible balanced bipartite graph, $M$ be a perfect matching in $G$, and $r \in V^+$.
    Then, any spanning $r$-in-arborescence of $G_M$ contains all the edges in $\overleftarrow{M}$.
\end{lemma}

\begin{proof}
Let $\Ain$ be a spanning $r$-in-arborescence of $G_M$.
Then, for any vertex $v \in V^-$, we have $\deg_{\Ain}^\mathrm{out}(v) = 1$. 
In $G_M = (V^+,V^-; \overrightarrow{E} \cup \overleftarrow{M})$, exactly one edge leaves each vertex in $V^-$, which is in $\overleftarrow{M}$.
Thus, $\Ain$ must contain all the edges in $\overleftarrow{M}$.
\end{proof}

In the previous algorithm, we essentially solve three problems: finding a minimum-weight perfect matching, a minimum-weight spanning in-arborescence, and a minimum-weight spanning out-arborescence.
Based on Lemma~\ref{lemma:matching in arborescence}, we merge the first two into a single task as follows: finding a minimum-weight spanning tree $T$ in $(G, w)$ among those having a perfect matching $M$ such that $\overrightarrow{T \setminus M} \cup \overleftarrow{M}$ is a spanning in-arborescence of $G_M$.

The main observation is the following characterization of such spanning trees.

\begin{lemma}\label{lemma:STPM iff PM and orientation}
  For a spanning tree $T$ of a balanced bipartite graph $G$, the following statements are equivalent.
  \begin{enumerate}
  \setlength{\itemsep}{0mm}
      \item $T$ has a perfect matching $M$ in $G$ and $\overrightarrow{T \setminus M} \cup \overleftarrow{M}$ is a spanning in-arborescence of $G_M$ (rooted at some vertex in $V^+$).
      \item There exists a vertex $r \in V^+$ such that $\deg_T(r) = 1$ and $\deg_T(v) = 2$ for every $v \in V^+ \setminus \{r\}$.
  \end{enumerate}
\end{lemma}
\begin{proof}
$(1 \Rightarrow 2)$~
Let $\Ain = \overrightarrow{T \setminus M} \cup \overleftarrow{M}$.
Since $M$ is a perfect matching in $G$, $\deg_{\Ain}^\mathrm{in}(v) = 1$ for every $v \in V^+$ and $\deg_{\Ain}^\mathrm{out}(v) = 1$ for every $v \in V^-$.
Also, since $\Ain$ is a spanning in-arborescence, $\deg_{\Ain}^\mathrm{out}(v) = 1$ for every $v \in V^+ \setminus \{r\}$ and $\deg_{\Ain}^\mathrm{out}(r) = 0$, where $r$ is the root of $\Ain$.
As $\deg_T(v) = \deg_{\Ain}^\mathrm{in}(v) + \deg_{\Ain}^\mathrm{out}(v)$ for every vertex $v$, we are done.

\medskip\noindent
$(2 \Rightarrow 1)$~
We first show that $T$ has a perfect matching.
By Hall's theorem~\cite{hall1935representatives}, it suffices to show that for any vertex set $X^+ \subseteq V^+$, we have $|X^+| \le |\Gamma_T(X^+)|$.
Suppose to the contrary that there exists a vertex set $X^+ \subseteq V^+$ such that $|X^+| > |\Gamma_T(X^+)|$.
Then, the subgraph formed by the edges in $T$ intersecting $X^+$ has at most $2|X^+| - 1$ vertices and at least $2|X^+| - 1$ edges (due to the degree condition), so it must contain a cycle, a contradiction.

Let $M$ be the perfect matching in $T$ (note that it is unique).
We next show that $\Ain = \overrightarrow{T \setminus M} \cup \overleftarrow{M}$ is an in-arborescence.
As $T$ is a tree, it suffices to show that $\deg_{\Ain}^\mathrm{out}(v) = 1$ for every vertex $v \in V \setminus \{r\}$, where $r \in V^+$ is the vertex specified in Statement~2.
By definition, every vertex $v \in V^-$ has a unique outgoing edge, which is in $\overleftarrow{M}$, and hence $\deg_{\Ain}^\mathrm{out}(v) = 1$.
Also, every vertex $v \in V^+$ has a unique incoming edge, which is in $\overleftarrow{M}$, and hence $\deg_{\Ain}^\mathrm{in}(v) = 1$.
Thus, for every vertex $v \in V^+ \setminus \{r\}$, we have $\deg_{\Ain}^\mathrm{out}(v) = \deg_T(v) - \deg_{\Ain}^\mathrm{in}(v) = 2 - 1 = 1$, and we are done.
\end{proof}

Let us call a spanning tree satisfying Statement~2 in Lemma~\ref{lemma:STPM iff PM and orientation} \emph{strongly balanced}, and say that the vertex $r$ is the \emph{root} and the tree is \emph{rooted at $r$}.
By Lemma~\ref{lemma: DMI includes arborescence}, every DM-irreducible balanced bipartite graph has a strongly balanced spanning tree, and then we formulate the first task as follows.

\begin{problem}[\textsc{Strongly Balanced Spanning Tree (SBST)}]
\mbox{ }\vspace{-1mm}
\begin{description}
\setlength{\itemsep}{0mm}
\item[Input:] A DM-irreducible balanced bipartite graph $G$ with $w \colon  E(G)  \to \mathbb{R}_{\ge 0}$.
\item[Goal:] Find a minimum-weight strongly balanced spanning tree of $G$.
\end{description}
\end{problem}

For the sake of simplicity of the following discussion, we consider the following problem, which is polynomial-time equivalent to the above problem and actually enough for our purpose.

\begin{problem}[\textsc{SBST with Root (SBSTR)}]
\mbox{ }\vspace{-1mm}
\begin{description}
\setlength{\itemsep}{0mm}
\item[Input:] A DM-irreducible balanced bipartite graph $G=(V^+,V^-;E)$ with $w \colon  E  \to \mathbb{R}_{\ge 0}$, and $r \in V^+$.
\item[Goal:] \mbox{Find a minimum-weight strongly balanced spanning tree of $G$ rooted at $r$.}
\end{description}
\end{problem}

Our proposed algorithm for \textsc{DMISS} is described as follows.
\begin{description}
\setlength{\itemsep}{0mm}
  \item[Step 1.]
  Choose an arbitrary vertex $r \in V^+$ as a root, and solve \textsc{SBSTR} with the input $(G, w, r)$.
  Let $T$ be the obtained solution, and $M$ be the perfect matching in $T$.
  \item[Step 2.]
  Find a minimum-weight spanning $r$-out-arborescence $\Aout$ in $(G_M, w_M)$ (recall that $w_M$ is defined as \eqref{eq:w_M}).
  \item[Step 3.]
  Output $T \cup \overline{\Aout}$.
\end{description}

The output is indeed a DM-irreducible spanning subgraph (by Lemma~\ref{lemma: DMI includes arborescence}).
To complete the proof of Theorem~\ref{thm:approx}, we analyze the computational time and the approximation ratio.

\subsubsection*{Computational Time}
Let $n = |V^+| = |V^-|$ and $m = |E|$.
As we already mentioned, the second task (finding a minimum-weight spanning arborescence) can be done in $O(m + n \log n)$ time~\cite{gabow1986efficient}.

The following observation (which was also observed in \cite{frank1998bipartite}) leads to an $O(n^3)$-time algorithm for the first task (solving \textsc{SBSTR}), which completes the analysis.
We omit the basics on matroids, for which we refer the readers to \cite{schrijver2003combinatorial}.

\begin{lemma}[cf.~{\cite[Lemma~5]{frank1998bipartite}}]\label{lemma:matroid intersection}
  \textsc{SBST} and \textsc{SBSTR} are formulated as the weighted matroid intersection problem, i.e., the set of feasible solutions can be represented as the set of common bases of two matroids.
\end{lemma}

\begin{proof}
We first introduce two matroids.
Let $\mathbf{M}_1 = (E, \mathcal{I}_1)$ be the cycle matroid of $G$, i.e.,
\[\mathcal{I}_1 = \{ F \subseteq E \mid \text{$F$ forms a forest (contains no cycle)} \},\]
and $\mathbf{M}_2 = (E, \mathcal{I}_2)$ be a partition matroid defined by
\[\mathcal{I}_2 = \{ F \subseteq E \mid \deg_F(v) \le 2 \ (\forall v \in V^+) \}.\]
Note that $V^+$ is a stable set in $G$ by definition.

In \textsc{SBST}, the set of feasible solutions is represented as the set of common bases of $\mathbf{M}_1$ and the truncation of $\mathbf{M}_2$ by $2|V^+| - 1$, whose independent set family is $\{ F \in \mathcal{I}_2 \mid |F| \le 2|V^+| - 1 \}$.
In \textsc{SBSTR}, the set of feasible solutions is represented as the set of common bases of $\mathbf{M}_1$ and a slightly modified partition matroid $\mathbf{M}'_2 = (E, \mathcal{I}'_2)$ such that
\[\mathcal{I}'_2 = \{ F \subseteq E \mid \deg_F(r) \le 1 \text{ and } \deg_F(v) \le 2 \ (\forall v \in V^+ \setminus \{r\}) \}. \qedhere\]
\end{proof}

Lemma~\ref{lemma:matroid intersection} implies that \textsc{SBST} and \text{SBSTR} can be solved in polynomial time\footnote{In contrast, it was recently shown that finding a strongly balanced spanning tree in a nonbipartite graph is NP-hard \cite{berczi2024finding}.} with the aid of matroid intersection algorithms.
In particular, Brezovec, Cornu{\'e}jols, and Glover~\cite{brezovec1988matroid} proposed an $O(\tilde{n}\tilde{m} + \tilde{k}\tilde{n}^2)$-time algorithm for the following special case:
one matroid is the cycle matroid of a graph and the other is a partition matroid whose partition classes are induced by a stable set of the graph (i.e., each partition class is the set of edges incident to a vertex in the stable set), where $\tilde{n}$ and $\tilde{m}$ are the numbers of vertices and of edges in the graph, respectively, and $\tilde{k}$ is the size of the stable set inducing the partition classes.
By the proof of Lemma~\ref{lemma:matroid intersection}, \textsc{SBSTR} completely fits this situation, and hence we can exactly solve it in $O(n^3)$ time, where we have $\tilde{k} = n = \frac{\tilde{n}}{2}$ and $\tilde{m} \le n^2$.

\begin{remark}
It seems not easy to essentially improve the computational time.
Both \textsc{SBST} and \textsc{SBSTR} actually include two fundamental special cases of the weighted matroid intersection problem: finding a minimum-weight perfect matching in a bipartite graph, and finding a minimum-weight spanning arborescence in a digraph.
In particular, for the former problem, one of the best-known bounds in general is $O(n(m + n \log n))$~\cite{edmonds1972theoretical}, which is not better than $O(n^3)$ when the graph is dense (i.e., $m = \Theta(n^2)$).

We sketch that \textsc{SBSTR} indeed includes the two problems.
For a balanced bipartite graph $G = (V^+, V^-; E)$, add two vertices $s^+, s^-$ and $|V^+| + 1$ edges $\{v, s^-\}$ $(v \in V^+ \cup \{s^+\})$ of weight $0$, and add a sufficiently large value (uniformly) to the weights of all the edges in $E$.
Then, a minimum-weight perfect matching in $G$ corresponds to a minimum-weight strongly balanced spanning tree of the resulting graph rooted at $s^+$, which must use all the additional edges of weight $0$.

Also, for a digraph $G = (V, E)$, split each vertex $v \in V$ into two vertices $v^+$ and $v^-$ with an edge $\{v^+, v^-\}$ of weight $0$, replace each edge $(u, v)$ with $\{u^+, v^-\}$ of the same weight plus a sufficiently large value (the same for all the edges).
Then, the resulting graph $G'$ is indeed a balanced bipartite graph, and a minimum-weight spanning in-arborescence of $G$ rooted at $r \in V$ corresponds to a minimum-weight strongly balanced spanning tree of $G'$ rooted at $r^+$, which must use all the edges $\{v^+, v^-\}$ of weight $0$.
\end{remark}

\subsubsection*{Approximation Ratio}
Let $\Sopt$ be an optimal solution.
We show $w(T \cup \overline{\Aout}) = w(T) + w_M(\Aout) \le 2w(\Sopt)$ (recall that $w_M \colon E(G_M) \to \mathbb{R}_{\ge 0}$ is defined as \eqref{eq:w_M}).
The following two claims imply this.

\begin{claim}\label{claim:spnning and opt}
  $w(T) \le w(\Sopt)$.
\end{claim}
\begin{proof}
By Lemmas~\ref{lemma: DMI includes arborescence} and \ref{lemma:matching in arborescence}, $\Sopt$ has a perfect matching $\Mopt$, and there exists a spanning $r$-in-arborescence $\Ain'$ of $G_{\Mopt}$ such that $\overleftarrow{\Mopt} \subseteq \Ain' \subseteq \overrightarrow{\Sopt} \cup \overleftarrow{\Mopt}$.
Since $T$ is optimal and $\overline{\Ain'}$ is feasible in \textsc{SBSTR} with the input $(G, w, r)$, we have $w(T) \le w(\overline{\Ain'}) \le w(\Sopt)$.
\end{proof}

\begin{claim}\label{claim: out-arorescence and opt}
  $w_M(\Aout) \le  w(\Sopt)$.
\end{claim}
\begin{proof}
As with the proof of Lemma~\ref{lemma: 3-approx}, consider a subgraph $G' = (V^+, V^-; \Sopt \cup M)$ of $G$.
Then, by Lemma~\ref{lemma: DMI includes arborescence}, $G'_M$ has a spanning $r$-out-arborescence $\Aout'$, and hence
\[w_M(\Aout) \le w_M(\Aout') = w(\overline{\Aout'} \setminus M) \le w(\Sopt).\qedhere\]
\end{proof}

\begin{figure}[tb]
  \centering
  \includegraphics[width=0.8\linewidth]{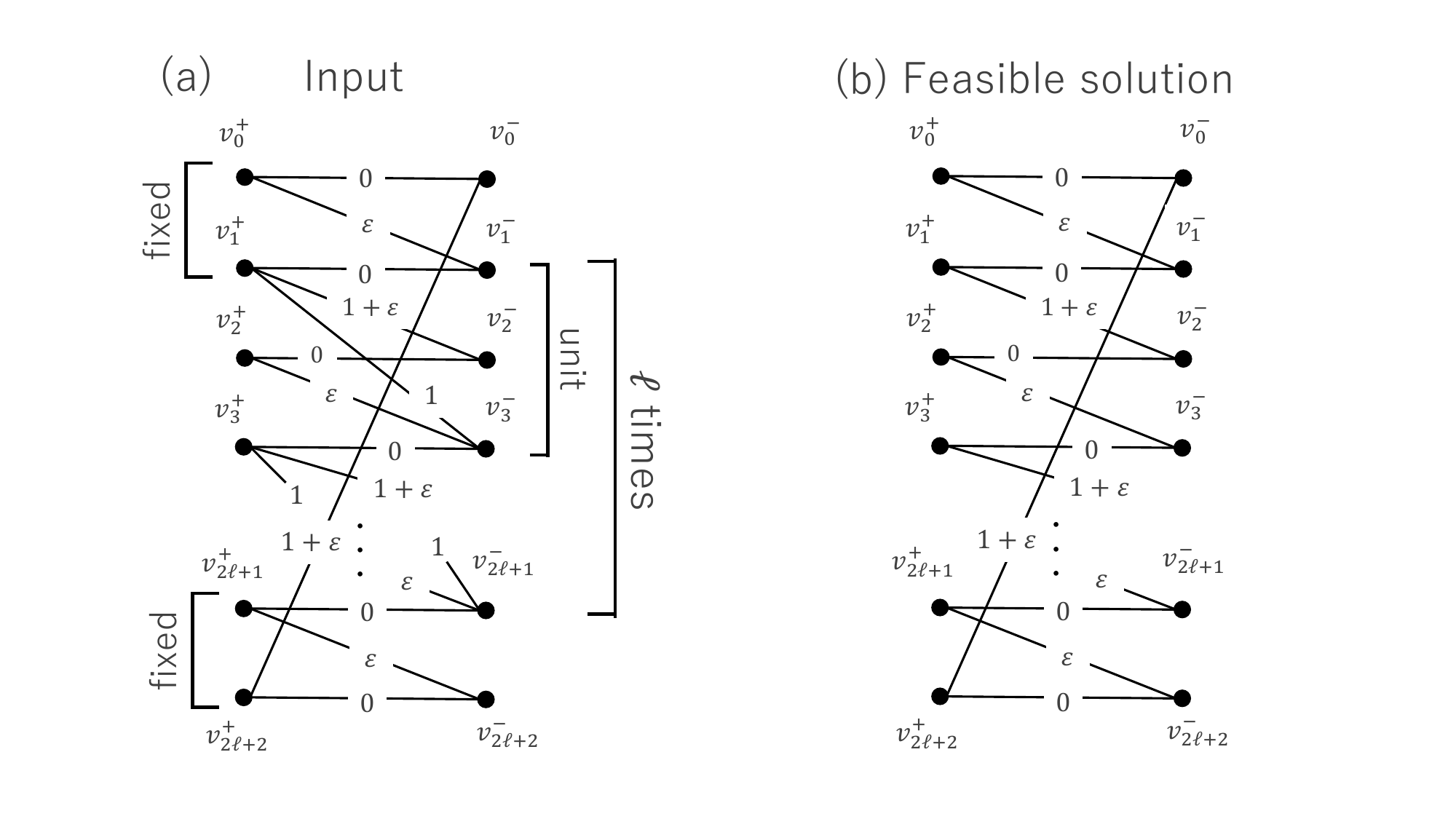}
  \caption{A tight example of our $2$-approximation algorithm for DMISS.}
  \label{fig1}
\end{figure}

\begin{remark}
The above analysis is tight.
Consider the input $(G, w)$ shown in Figure~\ref{fig1}, which consists of
\begin{itemize}
\setlength{\itemsep}{0mm}
\item $2(2\ell + 3)$ vertices $v_i^+, v_i^-$ $(i = 0, 1, \dots, 2\ell + 2)$,
\item $2\ell + 3$ edges $\{v_i^+, v_i^-\}$ $(i = 0, 1, \dots, 2\ell + 2)$ of weight $0$,
\item $\ell$ edges $\{v_{2i-1}^+, v_{2i}^-\}$ $(i = 1, 2, \dots, \ell)$ of weight $1 + \varepsilon$,
\item $\ell$ edges $\{v_{2i}^+, v_{2i+1}^-\}$ $(i = 1, 2, \dots, \ell)$ of weight $\varepsilon$,
\item $\ell$ edges $\{v_{2i-1}^+, v_{2i+1}^-\}$ $(i = 1, 2, \dots, \ell)$ of weight $1$,
\item two edges $\{v_0^+, v_1^-\}$ and $\{v_{2\ell+1}^+, v_{2\ell+2}^-\}$ of weight $\varepsilon$, and
\item an edge $\{v_{2\ell+2}^+, v_0^-\}$ of weight $1 + \varepsilon$,
\end{itemize}
where $\ell$ is a positive integer and $\varepsilon \coloneqq \frac{1}{\ell} > 0$.

While a Hamiltonian cycle $(v_0^-, v_0^+, v_1^-, v_1^+, \dots, v_{2\ell + 2}^-, v_{2\ell + 2}^+, v_0^-)$ illustrated as in (b) is a feasible solution, the output of our algorithm is the whole graph as follows.
First, the edges of weight at most $1$ form a spanning tree of $G$, which is of minimum weight (consider Kruskal's algorithm~\cite{kruskal1956shortest}).
This minimum-weight spanning tree is strongly balanced with root $r = v_{2\ell + 2}^+$, and hence it is eligible for $T$ in Step~1.
The edges of weight $0$ form a perfect matching $M$, and any spanning $r$-out-arborescence of $G_M$ (any candidate of $\Aout$ in Step~2) must contain all the remaining edges, which are of weight $1 + \varepsilon$ (as they are the only edges entering $v_{2i}$ $(i = 1, 2, \dots, \ell)$ in $G_M$).

Thus, the approximation ratio of our algorithm is at least
\[\frac{w(T \cup \overline{\Aout})}{w(\Sopt)} \ge \frac{\ell(2 + 2\varepsilon) + 2\varepsilon + 1 + \varepsilon}{\ell(1 + 2\varepsilon) + 2\varepsilon + 1 + \varepsilon} = \frac{2\ell^2 + 3\ell + 3}{\ell^2 + 3\ell + 3} \to 2 \quad (\ell \to \infty).\]
\end{remark}

\section{FPT Algorithm for Unweighted DMISS}\label{section:4}
In this section, we prove Theorem~\ref{thm:FPT} by presenting an FPT algorithm for \textsc{UnweightedDMISS}.

In Section~\ref{section:4.1}, as a base of our result, we sketch the idea of an FPT algorithm for \textsc{UnweightedSCSS} proposed by Bang-Jensen and Yeo~\cite{bang2008minimum}.
In Section~\ref{section:4.2}, by extending it, we propose an FPT algorithm for \textsc{UnweightedDMISS}.

\subsection{FPT Algorithm for Unweighted SCSS}\label{section:4.1}
The FPT algorithm utilizes another well-known characterization of strongly connected digraphs.
For a digraph $G = (V, E)$, an \emph{ear decomposition} of $G$ is a sequence $\mathcal{E} = (P_0, P_1, \dots, P_f)$ of edge-disjoint paths and cycles in $G$ with $E = \bigcup_{i=0}^f E(P_i)$ such that
\begin{itemize}
\setlength{\itemsep}{0mm}
\item $P_0$ is a cycle, and
\item $P_i$ $(i = 1, 2, \dots, f)$ is either a path between two different vertices in $\bigcup_{j=0}^{i-1} V(P_j)$ whose inner vertices are in $V \setminus \bigcup_{j=0}^{i-1} V(P_j)$, or a cycle intersecting exactly one vertex in $\bigcup_{j=0}^{i-1} V(P_j)$.
\end{itemize}
Each path or cycle $P_i$ is called an \emph{ear}.
An ear is called \emph{long} if it is of length at least $3$, where the length is evaluated by the number of edges.

\begin{lemma}[cf.~{\cite[Theorem~2.2.1]{frank2011connections}}] \label{lemma:SC no ear}
  A digraph is strongly connected if and only if it has an ear decomposition. 
\end{lemma} 

The idea of the FPT algorithm for \textsc{UnweightedSCSS} is sketched as follows.
It first computes an ear decomposition $\mathcal{E} = (P_0, P_1, \dots, P_f)$ of the input graph $G$ so that the set of long ears is inclusion-wise maximal (i.e., the set of long ears in $\mathcal{E}$ is not a proper subset of that in any ear decomposition $\mathcal{E}'$ of $G$, but the number of long ears in $\mathcal{E}$ is not necessarily maximized).
If the number of long ears is large enough (depending on the parameter $k$), then the ear decomposition $\mathcal{E}$ gives a strongly connected spanning subgraph that is small enough by removing all the ears of length $1$, and the answer is yes.
Otherwise, let $X$ be the set of vertices intersecting a long ear in $\mathcal{E}$.
Under the assumption that $\mathcal{E}$ does not immediately give a solution, we can show that $|X| = O(k)$.
Also, by the inclusion-wise maximality of long ears in $\mathcal{E}$, the vertices in $V \setminus X$ are added by ears of length $2$, and this is basically true for some minimum strongly connected spanning subgraph.
Thus, we can test the existence of a solution by an exhaustive search on the subgraphs of $G$ on the vertex set $X$ (plus a small subset), where we have only $2^{O(k^2)}$ candidates.

\subsection{FPT Algorithm for Unweighted DMISS}\label{section:4.2}
First, we give an alternative characterization of DM-irreducible balanced bipartite graphs by ear decompositions.
An ear decomposition of a bipartite graph $G = (V^+, V^-; E)$ is defined in the same way, and we call it \emph{odd proper} if
\begin{itemize}
\setlength{\itemsep}{0mm}
\item $P_0$ is a cycle of length at least $4$, and
\item each $P_i$ $(i = 1, 2, \dots, f)$ is a path of odd length, which implies that one end vertex of $P_i$ is in $V^+$ and the other is in $V^-$.
\end{itemize}

\begin{lemma}[cf.~{\cite[Theorem~2.4.17]{frank2011connections} or \cite[Theorem~4.1.6]{lovasz1986matching}}] \label{lemma:DMI no ear}
  A bipartite graph on at least four vertices is balanced and DM-irreducible if and only if it has an odd proper ear decomposition.
\end{lemma}

\begin{proof}
We give a proof to make this paper self-contained.

Suppose that a bipartite graph $G = (V^+, V^-; E)$ with $|V| \ge 4$ is balanced and DM-irreducible.
By Lemma~\ref{lemma: DM-irreducible}, there exists a perfect matching $M$ in $G$ and $G_M$ is strongly connected.
Since $|V| \ge 4$ and $G_M$ is connected, there exists an edge $e_0 = \{u_0, v_0\} \in E \setminus M$, where $u_0 \in V^+$ and $v_0 \in V^-$.
As $G_M$ is strongly connected, it contains a path from $v_0$ to $u_0$, which forms a cycle of length at least $4$ together with the edge $\overrightarrow{e_0} = (u_0, v_0)$.
Let $(P_0, P_1, \dots, P_f)$ be an arbitrary ear decomposition of $G_M$ such that $P_0$ is this cycle (the existence follows from Lemma~\ref{lemma:SC no ear} with shrinking $P_0$ into a single vertex).
We then see by induction on $i = 0, 1, \dots, f - 1$ that
\begin{itemize}
\setlength{\itemsep}{0mm}
\item for every $e = \{u, v\} \in M$, either $\{u, v\} \cap \bigcup_{j=0}^i V(P_j) = \emptyset$ or $\{u, v\} \subseteq \bigcup_{j=0}^i V(P_j)$ and $\overleftarrow{e} = (v, u) \in \bigcup_{j=0}^i E(P_j)$, and
\item $P_{i+1}$ starts and ends with edges in $\overrightarrow{E}$, which implies that it is of odd length.
\end{itemize}
Thus, we are done by taking the corresponding ear decomposition of the underlying graph (with removing each ear consisting of a single edge $\overrightarrow{e} \in \overrightarrow{M}$, which is duplicated).

Suppose that a bipartite graph $G$ with $|V(G)| \ge 4$ has an odd proper ear decomposition $(P_0, P_1, \dots, P_f)$.
We then see by induction on $i = 0, 1, \dots, f$ that the subgraph $G_i$ of $G$ defined by $V(G_i) \coloneqq \bigcup_{j=0}^i V(P_j)$ and $E(G_i) \coloneqq \bigcup_{j=0}^i E(P_j)$ is balanced and DM-irreducible (it is straightforward by Lemma~\ref{lemma: DM-irreducible}), which concludes that $G$ itself is balanced and DM-irreducible.
\end{proof}

An ear is called \emph{trivial} if it consists of a single edge, and \emph{nontrivial} otherwise.
Also, an ear is called \emph{long} if it is of length at least $5$.\footnote{This redefinition is reasonable by considering the relation between the strong connectivity of a digraph and the DM-irreducibility of a balanced bipartite graph (cf.~\cite[Section~2.3]{berczi2018DMImaking}); a path of length $\ell$ in the former corresponds to a path of length $2\ell - 1$ in the latter (as a vertex in the former is split into two adjacent vertices in the latter).}

We start to describe our FPT algorithm for \textsc{UnweightedDMISS}.
Let $G = (V^+, V^-; E)$ be the input DM-irreducible balanced bipartite graph, where $n = |V^+| \ge 2$ and $m = |E|$.
We first find an odd proper ear decomposition $\mathcal{E} = (P_0, P_1, \dots, P_f)$ such that
\begin{itemize}
\setlength{\itemsep}{1mm}
\item $\{P_1, P_2, \dots, P_s\}$ is the set of long ears in $\mathcal{E}$, which is inclusion-wise maximal (i.e., the set of long ears in $\mathcal{E}$ is not a proper subset of that in any ear decomposition $\mathcal{E}'$ of $G$, but $s$ is not necessarily maximized), and
\item $\{P_1, P_2, \dots, P_r\}$ is the set of nontrivial ears in $\mathcal{E}$ (note that $s \le r \le f$).
\end{itemize}
This can be done in polynomial time.\footnote{For example, naively, we can start with an arbitrary cycle $P_0$ of length at least $4$, and greedily add a long ear as far as possible; a long ear can be found by an exhaustive search of the first and last three vertices, where the number of candidates is $O(n^6)$.}

For each $i = 0, 1, \dots, f$, we define a subgraph $G_i$ of $G$ by $V(G_i) \coloneqq \bigcup_{j=0}^i V(P_j)$ and $E(G_i) \coloneqq \bigcup_{j=0}^i E(P_j)$.
Let $X = V(G_s)$ and $Y = V \setminus X$.
We then observe the following two claims.

\begin{claim}\label{claim:DMISS maximum}
  $G_r$ is a DM-irreducible spanning subgraph of $G$, which has $(3n-2)-\left(\frac{|X|}{2}-s-2\right)$ edges.
\end{claim}
\begin{proof}
    Since $V(G_r) = V(G)$, by Lemma~\ref{lemma:DMI no ear}, $G_r$ is a DM-irreducible spanning subgraph of $G$.
    By definition, we have
\begin{align*}
|V(G_r)| &= |V(P_0)| + \sum_{i=1}^r (|V(P_i)| - 2) = |V(P_0)| + \sum_{i=1}^r (|E(P_i)| - 1),\\
|E(G_r)| &= |E(P_0)| + \sum_{i=1}^r |E(P_i)| = |V(P_0)| + \sum_{i=1}^r |E(P_i)|,
\end{align*}
    and hence $|E(G_r)| = |V(G_r)| + r = 2n + r = (3n - 2) - (n - r - 2)$.
    We also have
    \[|X| = 2n - |Y| = 2n - \sum_{i=s+1}^r (|E(P_i)| - 1) = 2n - 2(r - s),\]
    and we are done.
\end{proof}

By Claim~\ref{claim:DMISS maximum}, if $k \le \frac{|X|}{2}-s-2$, then $(G, k)$ is a yes-instance.
We also have $|X| \ge 4s + 4$ by definition of long ears.
In what follows, we assume $k \ge \frac{|X|}{2} - s - 1 \ge \frac{|X|}{4}$, i.e., $|X| \le 4k$.

\begin{claim}\label{claim: perfect matching}
$E(G[Y])$ is a perfect matching in $G[Y]$.
\end{claim}

\begin{proof}
By definition, $P_{s+1}, P_{s+2}, \dots, P_r$ are paths of length $3$ and the two inner vertices of each path are in $Y$.
Thus, $E(G[Y])$ indeed has a perfect matching
\[\{\{y_i^+, y_i^-\} \mid \text{$\{y_i^+, y_i^-\}$ is the middle edge of the ear $P_i$ $(s + 1 \le i \le r)$}\}.\]
By the maximality of long ears, we can see that $E(G[Y])$ contains no other edge as follows.

Suppose to the contrary that $E(G[Y])$ contains another edge $e = \{y_i^+, y_j^-\}$ with $i \neq j$.
We pick such an edge minimizing $i + j$, and assume $i < j$ by symmetry.
Then, $G_j$ has a path $Q_i$ of length $2$ between $y_i^+$ and some vertex $x^+ \in X \cap V^+$ (a part of the ear $P_i$) and a path $Q_j$ of length at least $2$ between $y_j^-$ and some vertex $x^- \in X \cap V^-$ starting with the edge $\{y_j^+, y_j^-\}$ whose inner vertices are in $Y$.
If $Q_i$ and $Q_j$ are disjoint, then $Q_i$, $Q_j$, and $e$ form a long ear that can be added to $G_s$ (just after $P_s$ in $\mathcal{E}$), which contradicts the maximality of long ears.
Otherwise, $Q_j$ start with the edges $\{y_j^+, y_j^-\}$ and $\{y_j^+, y_i^-\}$ (otherwise, there exists another edge $e' = \{y_{i'}^+, y_{j'}^-\}$ with $i' \neq j'$ and $i' + j' < i + j$, a contradiction).
In this case, the first two edges of $Q_j$, $e$, and the two end edges in the ear $P_i$ form a long ear that can be added to $G_s$ (just after $P_s$ in $\mathcal{E}$), which again contradicts the maximality of long ears.
\end{proof}

We consider a bipartite graph $\tilde{H}$ with bipartition $(\tilde{Z}, \tilde{Y})$ defined as follows:
\begin{align*}
\tilde{Y} &\coloneqq \{ \tilde{y} \coloneqq\{y^+,y^-\} \mid y^+,y^- \in Y,\ \{y^+,y^-\} \in E \},\\
\tilde{Z} &\coloneqq \{ \tilde{z} \coloneqq\{ z^+, z^- \} \mid \exists \tilde{y} = \{y^+, y^-\} \in \tilde{Y},\ \{\{ z^+, y^- \}, \{z^-, y^+\}\} \subseteq E \},\\
E(\tilde{H}) &\coloneqq \{ \{\tilde{z}, \tilde{y} \} \mid \tilde{z} = \{z^+, z^-\} \in \tilde{Z},\  \tilde{y} = \{y^+, y^-\} \in \tilde{Y},\ \{\{ z^+, y^- \}, \{z^-, y^+\}\} \subseteq E \}.
\end{align*}
Note that each vertex in $\tilde{H}$ is a pair of vertices in $G$, one of which is in $V^+$ and the other is in $V^-$.

Let $\Zout \subseteq \tilde{Z}$ be a vertex set maximizing $|\Zout|$ subject to $2|\Zout| > |\Gamma_{\tilde{H}} (\Zout)|$, and let $\Zin \coloneqq \tilde{Z} \setminus \Zout$.
Then, by the maximality of $\Zout$, we have $2|\tilde{Z}'| \le |\Gamma_{\tilde{H}} (\tilde{Z}')|$ for any $\tilde{Z}' \subseteq \Zin$.
Let $\Yout \coloneqq \Gamma_{\tilde{H}} (\Zout) \setminus \Gamma_{\tilde{H}} (\Zin)$, $\Yin \coloneqq \tilde{Y} \setminus \Yout$, and $\Hin \coloneqq \tilde{H}[\Zin \cup \Yin]$.
Then, as $\Gamma_{\tilde{H}}(\Zin) \subseteq \Yin$, we have $2|\tilde{Z}'| \le |\Gamma_{\tilde{H}} (\tilde{Z}')| = |\Gamma_{\Hin} (\tilde{Z}')|$ for any $\tilde{Z}' \subseteq \Zin$.
By Hall's theorem~\cite{hall1935representatives}, $\Hin$ has two disjoint matchings such that each of them covers $\Zin$ and no vertex in $\Yin$ is covered by both of them.
In other words, $\Hin$ contains a subgraph $\tilde{H}'$ such that $\deg_{\tilde{H}'}(\tilde{z}) = 2$ for every $\tilde{z} \in \Zin$ and $\deg_{\tilde{H}'}(\tilde{y}) \leq 1$ for every $\tilde{y} \in \Yin$.
Let $\tilde{M} \coloneqq E(\tilde{H}')$, and $\tilde{Z}(\tilde{M})$ and $\tilde{Y}(\tilde{M})$ be the sets of end vertices of edges in $\tilde{M}$ that are in $\tilde{Z}$ and in $\tilde{Y}$, respectively.

If $\Zin = \emptyset$, then $\Yout = \tilde{Y}$.
Hence, $|\tilde{Y}| = |\Yout| < 2|\Zout| \le 2|\tilde{Z}| \le 2(2k)^2 = 8k^2$ (recall that $|X| \le 4k$ and each vertex $\tilde{z} \in \tilde{Z}$ is a pair of vertices in $G$ one of which is in $X \cap V^+$ and the other is in $X \cap V^-$).
This implies $|V| = |X| + |Y| = |X| + 2|\tilde{Y}| \le 4k(4k+1)$.
In this case, the input size is bounded by a function of the parameter $k$, and hence we can solve the problem by an exhaustive search in FPT time.

Otherwise, i.e., if $\Zin \neq \emptyset$, the following claim holds, which is analogous to \cite[Claim in p.~2927]{bang2008minimum}.
We define $\tilde{Y}' \coloneqq \Yin \setminus \tilde{Y}(\tilde{M})$ and $Y' \coloneqq \bigcup_{\tilde{y} \in \tilde{Y}'} \tilde{y} \ (\subseteq V)$.

\begin{claim}\label{claim:DMISS claim}
There exists a minimum DM-irreducible spanning subgraph $G^*$ of $G$ satisfying the following two conditions.
\begin{itemize}
\setlength{\itemsep}{0mm}
\item $G^* - Y'$ is a DM-irreducible spanning subgraph of $G - Y'$.
\item For any $\tilde{y} = \{y^+, y^-\} \in \tilde{Y}'$, there exists $\tilde{z} = \{z^+, z^-\} \in \tilde{Z}$ with $\{\tilde{z}, \tilde{y}\} \in E(\tilde{H})$ such that the edges in $E(G^*)$ incident to $y^+$ or $y^-$ are $\{z^+, y^-\}$, $\{y^+, y^-\}$, and $\{y^+, z^-\}$.
\end{itemize}
\end{claim}

Before proving Claim~\ref{claim:DMISS claim}, we complete our FPT algorithm for \textsc{UnweightedDMISS}.
By Claim~\ref{claim:DMISS claim}, the answer is yes if and only if there exists a DM-irreducible spanning subgraph of $G - Y'$ consisting of at most $3n - 2 - k - 3|\tilde{Y}'|$ edges.
We can check all possible spanning subgraphs of $G - Y'$ by an exhaustive search in FPT time, since
\[|V(G - Y')| = |X| + 2|\tilde{Y}(\tilde{M})| + 2|\tilde{Y}_\mathrm{out}| < |X| + 4|\tilde{Z}(\tilde{M})| + 4|\tilde{Z}_\mathrm{out}| = |X| + 4|\tilde{Z}| \le 4k(4k+1).\]

\begin{proof}[Proof of Claim~\ref{claim:DMISS claim}]
Let $G'$ be a minimum DM-irreducible spanning subgraph of $G$.
If $G' - Y'$ is DM-irreducible, then we are done as follows.
By Claim~\ref{claim: perfect matching}, any DM-irreducible spanning subgraph of $G$ (including $G'$) contains, for each pair $\tilde{y} = \{y^+, y^-\} \in \tilde{Y}$, either
\begin{itemize}
\setlength{\itemsep}{0mm}
\item a path of length $3$ between two vertices in $V(G' - Y)$ through the edge $\{y^+, y^-\}$, or
\item at least two edges incident to $y^+$ and to $y^-$ other than $\{y^+, y^-\}$ (at least four in total).
\end{itemize}
Thus, a spanning subgraph of $G$ obtained from $G' - Y'$ as follows is DM-irreducible and has as few edges as $G'$ itself, and is a desired graph: for each pair $\tilde{y} = \{y^+, y^-\} \in \tilde{Y}'$, choose an arbitrary pair $\tilde{z} = \{z^+, z^-\} \in \tilde{Z}$ with $\{\tilde{z}, \tilde{y}\} \in E(\tilde{H})$ and add three edges $\{z^+, y^-\}$, $\{y^+, y^-\}$, and $\{y^+, z^-\}$.

Otherwise, i.e., if $G' - Y'$ is not DM-irreducible, we transform $G'$ as follows: remove all the edges incident to $Y(\tilde{M}) \coloneqq \bigcup_{\tilde{y} \in \tilde{Y}(\tilde{M})}\tilde{y} \ (\subseteq V)$, and for each $\{\tilde{z}, \tilde{y}\} \in \tilde{M}$, add three edges $\{z^+, y^-\}$, $\{y^+, y^-\}$, and $\{y^+, z^-\}$.
Let $G^\sharp$ be the resulting graph.
We then have $|E(G^\sharp)| \le |E(G')|$, and it suffices to show that $G^\sharp - Y'$ is DM-irreducible.

Let $M'$ be a perfect matching in $G'$.
As $G'$ is DM-irreducible, $G'_{M'}$ is strongly connected.
We first show that $G^\sharp - Y'$ has a perfect matching, which we construct from $M'$ minus the edges incident to $Y_\mathrm{in} \coloneqq \bigcup_{\tilde{y} \in \tilde{Y}_\mathrm{in}}\tilde{y} \ (\subseteq V)$ as follows.
For each removed edge $\{y^+, y^-\} \in E(G[Y(\tilde{M})]) \cap M'$, we just add it again.
The other removed edges in $M'$ are paired as $\{y^+, z^-\}$ and $\{z^+, y^-\}$ such that $\tilde{y} = \{y^+, y^-\} \in \Yin$, $\tilde{z} = \{z^+, z^-\} \in \Zin$, and $\{\tilde{z}, \tilde{y}\} \in E(\Hin)$.
By definition of $\tilde{M}$, there exists a pair $\tilde{y}' = \{y'^+, y'^-\} \in \tilde{Y}(\tilde{M})$ with $\{\tilde{z}, \tilde{y}'\} \in \tilde{M} \subseteq E(\Hin)$, which means that $\{z^+, y'^-\}, \{y'^+, y'^-\}, \{y'^+, z^-\} \in E(G^\sharp)$.
By definition of $\tilde{M}$ again, such $\tilde{y}'$ is distinct for each $\tilde{z} \in \Zin = \tilde{Z}(\tilde{M})$, and hence we can add $\{z^+, y'^-\}$ and $\{y'^+, z^-\}$ instead of each pair of $\{y^+, z^-\}$ and $\{z^+, y^-\}$.
Finally, for each pair $\{y^+, y^-\} \in \tilde{Y}(\tilde{M})$, if at least one of $y^+$ and $y^-$ is not matched, then both $y^+$ and $y^-$ are not matched due to the above procedure, and add an edge $\{y^+, y^-\}$.
Let $M^\sharp$ be the obtained perfect matching in $G^\sharp$.

The remaining task is to show that $G^\sharp_{M^\sharp} - Y'$ is strongly connected.
By the above construction, it suffices to show that for each pair of two vertices $s, t \in V \setminus Y_\mathrm{in}$, $t$ is reachable from $s$ in $G^\sharp_{M^\sharp} - Y'$.
Suppose to the contrary that there exists a pair $(s, t)$ such that $t$ is not reachable from $s$ in $G^\sharp_{M^\sharp} - Y'$.
As $G'_{M'}$ is strongly connected, there exists a path $P$ from $s$ to $t$ in $G'_{M'}$.
Take such a pair of $(s, t)$ and $P$ so that the length of $P$ is minimized.
Then, all the inner vertices of $P$ are in $Y_\mathrm{in}$ (otherwise, we can take an inner vertex $v \not\in Y_\mathrm{in}$ such that a pair of $(s, v)$ and the prefix of $P$ or a pair of $(v, t)$ and the suffix of $P$ is eligible, which contradicts that the length of $P$ is minimized), and hence $s, t \in Z_\mathrm{in}$.

If $\{s, t\} \in \Zin = \tilde{Z}(\tilde{M})$, then by the above construction, there exists a path from $s$ to $t$ through a pair $\tilde{y} \in \tilde{Y}(\tilde{M})$ in $G^\sharp_{M^\sharp} - Y'$, a contradiction.
Otherwise, since $E(G[Y])$ is a matching (Claim~\ref{claim: perfect matching}), $P$ consists of two edges $(s, v)$ and $(v, t)$ such that $v \in Y_\mathrm{in}$ and exactly one of $\{s, v\}$ and $\{t, v\}$ is in $M'$.
By symmetry, assume that $s, t \in V^+$, $v \in V^-$, and $\{t, v\} \in M'$.
Let $\tilde{y} = \{y^+, y^-\} \in \Yin$ be the pair with $y^- = v$.
Then, there exist two pairs $\tilde{z}_1 = \{z_1^+, z_1^-\}$ and $\tilde{z}_2 = \{z_2^+, z_2^-\}$ in $\Zin = \tilde{Z}(\tilde{M})$ such that $z_1^+ = s$, $z_2^+ = t$, $z_1^- = z_2^-$ and $\{y^+, z_1^-\} \in M'$.
By the above construction, there exist two pairs $\tilde{y}_1 = \{y_1^+, y_1^-\}$ and $\tilde{y}_2 = \{y_2^+, y_2^-\}$ in $\tilde{Y}(\tilde{M})$ such that $\{\tilde{z}_1, \tilde{y}_1\}, \{\tilde{z}_2, \tilde{y}_2\} \in \tilde{M}$ and $\{y_1^+, y_1^-\}, \{z_2^+, y_2^-\}, \{y_2^+, z_2^-\} \in M^\sharp$.
This concludes that $G^\sharp_{M^\sharp} - Y'$ has a path from $s = z_1^+$ to $t = z_2^+$ that consists of six edges $(s, y_1^-)$, $(y_1^-, y_1^+)$, $(y_1^+, z_1^-)$, $(z_2^-, y_2^+)$, $(y_2^+, y_2^-)$, and $(y_2^-, t)$, a contradiction.
Thus, we are done.
\end{proof}

\begin{remark}
Another possible parameterized approach is, as a connectivity augmentation problem (like \cite{berczi2018DMImaking}), taking the parameter $k$ as the number of additional edges to achieve DM-irreducibility.
For strong connectivity augmentation, Klinkby, Misra, and Saurabh~\cite{klinkby2021strong} recently showed the first FPT algorithm in this direction.
How to extend their result to the DM-irreducibility problem is highly nontrivial because there are too many possible choices of perfect matchings in the resulting augmented graphs.
This may be possible future work.
\end{remark}

\bibliographystyle{plain}
\bibliography{main}

\end{document}